\documentclass[a4paper,USenglish,cleveref]{lipics-v2019}
\pdfoutput=1

\usepackage{amsmath}
\usepackage{hyperref}
\usepackage{tikz}
\usepackage{mathpartir}

\usetikzlibrary{cd}

\newcommand{\mycat}[1]{\mathcal{#1}}
\newcommand{\myicat}[1]{\mathbf{#1}}
\newcommand{\mycatconst}[1]{\mathbf{#1}}
\newcommand{\mysconst}[1]{\mathsf{#1}}
\DeclareMathOperator{\myop}{op}
\newcommand{\mySet}{\mycatconst{Set}}
\newcommand{\myPrshv}{\mycatconst{PSh}}
\newcommand{\myAsm}{\mycatconst{Asm}}
\newcommand{\myPER}{\mycatconst{PER}}
\newcommand{\myrefl}{\mysconst{refl}}
\newcommand{\myind}{\mysconst{ind}}
\newcommand{\myinl}{\mysconst{inl}}
\newcommand{\myinr}{\mysconst{inr}}
\newcommand{\mymin}{\mysconst{min}}
\newcommand{\mymax}{\mysconst{max}}
\newcommand{\myProp}{\mysconst{Prop}}
\newcommand{\myhProp}{\mysconst{hProp}}
\newcommand{\myId}{\mysconst{Id}}
\newcommand{\myPath}{\mysconst{Path}}
\newcommand{\mytuple}[1]{\langle #1 \rangle}
\newcommand{\myfromto}{\leftrightarrow}
\newcommand{\myN}{\mathbb{N}}
\newcommand{\myS}{\mathbb{S}}
\newcommand{\myyoneda}{\mathbf{y}}
\newcommand{\myUnittype}{\mathbf{1}}
\newcommand{\myEmptytype}{\mathbf{0}}
\newcommand{\myBoolean}{\mathbf{2}}
\newcommand{\myI}{\mathbb{I}}
\newcommand{\myUniverse}{\mathcal{U}}
\newcommand{\myinterval}{\myI}
\newcommand{\myel}{\mysconst{el}}
\newcommand{\myCof}{\mysconst{Cof}}
\newcommand{\myComp}{\mysconst{Comp}}
\newcommand{\myCComp}{\mysconst{C}}
\newcommand{\myFib}{\mysconst{Fib}}
\newcommand{\myGlue}{\mysconst{Glue}}
\newcommand{\mySGlue}{\mysconst{SGlue}}
\newcommand{\mycomp}{\mysconst{comp}}
\newcommand{\myid}{\mysconst{id}}
\newcommand{\myprsv}{\mysconst{pres}}
\newcommand{\mytransport}{\mysconst{tp}}
\newcommand{\mylift}{\mysconst{iea}}
\newcommand{\myequiv}{\mysconst{equiv}}
\newcommand{\myFin}{\mysconst{Fin}}
\newcommand{\myord}{\mysconst{ord}}
\newcommand{\mycatEl}{\mycatconst{El}}

\theoremstyle{plain}
\newtheorem{axiom}[theorem]{Axiom}

\Crefname{axiomi}{Axiom}{Axioms}

\hypersetup{
 pdfauthor={Taichi Uemura},
 pdftitle={Cubical Assemblies, a Univalent and Impredicative
  Universe and a Failure of Propositional Resizing},
 pdflang={English}
}

\nolinenumbers

\title{Cubical Assemblies, a Univalent and Impredicative
  Universe and a Failure of Propositional Resizing}
\titlerunning{Cubical Assemblies}
\author{Taichi Uemura}
{University of Amsterdam, Amsterdam, the Netherlands}
{t.uemura@uva.nl}
{https://orcid.org/0000-0003-4930-1384}{}
\authorrunning{T. Uemura}
\Copyright{Taichi Uemura}
\funding{This work is part of the research programme ``The Computational
Content of Homotopy Type Theory'' with project number 613.001.602,
which is financed by the Netherlands Organisation for Scientific
Research (NWO).}
\acknowledgements{I would like to thank Benno van den Berg, Martijn
  den Besten and Andrew Swan for helpful discussions and comments, and
  Bas Spitters, Steve Awodey and the anonymous reviewer
  for their comments, questions and suggestions.}

\hideLIPIcs

\begin{document}

\maketitle

\begin{abstract}
  We construct a model of cubical type theory with a univalent and
  impredicative universe in a category of cubical assemblies. We show
  that this impredicative universe in the cubical assembly model does
  not satisfy a form of propositional resizing.
\end{abstract}

\section{Introduction}
\label{sec:org60aa1d6}
\emph{Homotopy type theory} \cite{hottbook} is an extension of
Martin-L\"{o}f's dependent type theory
\cite{martin-lof1975intuitionistic} with homotopy-theoretic ideas. The
most important features are Voevodsky's \emph{univalence axiom} and
\emph{higher inductive types} which provide a novel synthetic way of
proving theorems of abstract homotopy theory and formalizing
mathematics in computer proof assistants \cite{bauer2017library}.

Ordinary homotopy type theory \cite{hottbook} uses a cumulative hierarchy of
universes
\[\myUniverse_{0} : \myUniverse_{1} : \myUniverse_{2} : \dots,\]
but there is another choice of universes: one \emph{impredicative}
universe in the style of the Calculus of Constructions
\cite{coquand1988calculus}. Here we say a universe \(\myUniverse\) is
impredicative if it is closed under dependent products along any
type family: for any type \(A\) and function \(B : A \to
  \myUniverse\), the dependent product \(\prod_{x : A}B(x)\) belongs
to \(\myUniverse\). An interesting use of such an impredicative universe
in homotopy type theory is the \emph{impredicative encoding} of higher
inductive types, proposed by Shulman \cite{shulman2011impredicative},
as well as ordinary inductive types in polymorphic type theory
\cite{girard1989proofs}. For instance, the unit circle
\(\myS^{1}\) is encoded as \(\prod_{X : \myUniverse}\prod_{x : X}x =
  x \to X\) which has a base point and a loop on the point and
satisfies the recursion principle in the sense of the HoTT book
{\cite[Chapter 6]{hottbook}}. Although the impredicative encoding of a higher
inductive type does not satisfy the induction principle in general,
some truncated higher inductive types have refinements of the
encodings satisfying the induction principle
\cite{speight2017impredicative,awodey2018impredicative}.

In this paper we construct a model of type theory with a univalent
and impredicative universe to prove the consistency of that type
theory. Impredicative universes are modeled in the category of
\emph{assemblies} or \emph{\(\omega\)-sets}
\cite{longo1991constructive,phoa2006introduction}, while univalent
universes are modeled in the categories of groupoids
\cite{hofmann1998groupoid}, simplicial sets
\cite{kapulkin2016simplicialv4} and cubical sets
\cite{bezem2014model,bezem2017univalence}. Therefore, in order to
construct a univalent and impredicative universe, it is natural to
combine them and construct a model of type theory in the category of
internal groupoids, simplicial or cubical objects in the category of
assemblies.  There has been an earlier attempt to obtain a univalent
and impredicative universe by Stekelenburg
\cite{stekelenburg2016constructive} who took a simplicial approach. A
difficulty with this approach is that the category of assemblies
does not satisfy the axiom of choice or law of excluded middle, so
it becomes harder to obtain a model structure on the category of
simplicial objects. Another approach is taken by van den Berg
\cite{vandenberg2018univalent} using groupoid-like objects, but his
model has a dimension restriction. Our choice is the cubical objects
in the category of assemblies, which we will call \emph{cubical
assemblies}. Since the model in cubical sets
\cite{bezem2014model,cohen2016cubical} is expressed, informally, in a
constructive metalogic, one would expect that their construction can
be translated into the internal language of the category of
assemblies. A similar approach is taken by Awodey, Frey and Hofstra
\cite{awodey2017impredicative,frey2017realizability}.

Instead of a model of homotopy type theory itself, we construct a
model of a variant of \emph{cubical type theory} \cite{cohen2016cubical} in
which the univalence axiom is provable. Orton and Pitts
\cite{orton2016axioms} gave a sufficient condition for modeling
cubical type theory without universes of fibrant types in an
elementary topos equipped with an interval object \(\myI\). Although
the category of cubical assemblies is not an elementary topos, most
of their proofs work in our setting because they use a dependent
type theory as an internal language of a topos and the category of
cubical assemblies is rich enough to interpret the type theory. For
construction of the universe of fibrant types, we can use the right adjoint to the
exponential functor \((\myI \to -)\) in the same way as Licata, Orton,
Pitts and Spitters \cite{licata2018internal}.

Voevodsky \cite{voevodsky2012polymorphic} has proposed the
\emph{propositional resizing axiom} {\cite[Section 3.5]{hottbook}}
which implies that every homotopy proposition is equivalent to some
homotopy proposition in the smallest universe. The propositional
resizing axiom can be seen as a form of impredicativity for homotopy
propositions. Since the universe in the cubical assembly model is
impredicative, one might expect that the cubical assembly model
satisfies the propositional resizing axiom. Indeed, for a homotopy
proposition \(A\), we have an approximation \(A^{*}\) of \(A\) by a
homotopy proposition in \(\myUniverse\) defined as
\[
  A^{*} := \prod_{X : \myhProp}(A \to X) \to X,
\]
where \(\myhProp\) is the universe of homotopy propositions in
\(\myUniverse\), and \(A\) is equivalent to some homotopy proposition
in \(\myUniverse\) if and only if the function
\(\lambda aXh.ha : A \to A^{*}\) is an equivalence. However, the
propositional resizing axiom fails in the cubical assembly model. We
construct a homotopy proposition \(A\) such that the function \(A \to
A^{*}\) is not an equivalence.

We begin \cref{sec:org52feb89} by formulating the axioms
for modeling cubical type theory given by Orton and Pitts
\cite{orton2016axioms,orton2017axioms} in a weaker setting. In
\cref{sec:orgb4d8a75} we describe how to construct a model of cubical type
theory under those axioms. In
\cref{sec:org3425270} we give a sufficient condition for presheaf
models to satisfy those axioms. As an example of presheaf
model we construct a model of cubical type theory in cubical
assemblies in \cref{sec:orga3e7717}, and show that the
cubical assembly model does not satisfy the propositional resizing
axiom.

\section{The Orton-Pitts Axioms}
\label{sec:org52feb89}
We will work in a model \(\mycat{E}\) of dependent type theory with
\begin{itemize}
\item dependent product types, dependent sum types, extensional identity
types, unit type, disjoint finite coproducts and propositional
truncation;
\item a constant type \({} \vdash \myinterval\), called an
\emph{interval}, with two constants \({} \vdash 0 : \myinterval\) and
\({} \vdash 1 : \myinterval\) called \emph{end-points} and two
operators \(i, j : \myinterval \vdash i \sqcap j : \myinterval\)
and \(i, j : \myinterval \vdash i \sqcup j : \myinterval\) called
\emph{connections};
\item a dependent right adjoint to the exponential functor \((\myI \to
    -) : \mycat{E} \to \mycat{E}\);
\item a propositional universe \({} \vdash \myCof\) whose inhabitants
are called \emph{cofibrations};
\item an impredicative universe \({} \vdash \myUniverse\)
\end{itemize}
satisfying the axioms listed in \cref{orgda6ce3b}.
\begin{figure}
\begin{enumerate}
\item \label[axiomi]{org09c88da}
\(\neg (0 = 1)\)
\item \label[axiomi]{org541f57d}
\(\forall_{i : \myI}0 \sqcap i = i \sqcap 0 = 0
     \land 1 \sqcap i = i \sqcap 1 = i\)
\item \label[axiomi]{org6fd8c58}
\(\forall_{i : \myI}0 \sqcup i = i \sqcup 0 = i
     \land 1 \sqcup i = i \sqcup 1 = 1\)
\item \label[axiomi]{orge92c1ca}
\(i : \myI \vdash i = 0 : \myCof\)
\item \label[axiomi]{orgc87da6c}
\(i : \myI \vdash i = 1 : \myCof\)
\item \label[axiomi]{org629a444}
\(\varphi, \psi : \myCof \vdash \varphi \lor \psi : \myCof\)
\item \label[axiomi]{org8e1319f}
\(\varphi : \myCof, \psi : \varphi \to \myCof \vdash
     \sum_{u : \varphi}\psi u : \myCof\)
\item \label[axiomi]{org4e8c5eb}
\(\varphi : \myI \to \myCof \vdash \forall_{i : \myI}\varphi i :
     \myCof\)
\item \label[axiomi]{org46f619f}
\(\forall_{\varphi, \psi : \myCof}(\varphi \myfromto \psi) \to (\varphi = \psi)\)
\item \label[axiomi]{org0fe33f4}
\(\varphi : \myCof, A : \varphi \to \myUniverse, B :
      \myUniverse, f : \prod_{u : \varphi}Au \cong B \vdash
      \mylift(\varphi, f) : \sum_{\bar{A} : \myUniverse}\{\bar{f} :
      \bar{A} \cong B \mid \forall_{u : \varphi}(Au, fu) = (\bar{A},
      \bar{f})\}\)
\end{enumerate}
\caption{\label{orgda6ce3b}
The Orton-Pitts Axioms}
\end{figure}
In the rest of the section we explain these conditions in more detail.

The dependent type theory we use is Martin-L\"{o}f's extensional
type theory \cite{martin-lof1975intuitionistic}. The notion of model
of dependent type theory we have in mind is categories with families
\cite{dybjer1996internal} equipped with certain algebraic operators
corresponding to the type formers. A \emph{category with families}
\(\mycat{E}\) consists of:
\begin{itemize}
\item a category \(\mycat{E}\) of \emph{contexts} with a terminal object
denoted by \({\cdot}\);
\item a presheaf \(\Gamma \mapsto \mycat{E}(\Gamma) : \mycat{E}^{\myop}
    \to \mySet\) of \emph{types};
\item a presheaf \((\Gamma, A) \mapsto \mycat{E}(\Gamma \vdash A) :
    \mycatEl(\mycat{E}(-))^{\myop} \to \mySet\) of \emph{terms}, where
\(\mycatEl(P)\) is the category of elements for a presheaf \(P\)
\end{itemize}
such that, for any context \(\Gamma \in \mycat{E}\) and type
\(A \in \mycat{E}(\Gamma)\), the presheaf
\[(\mycat{E}/\Gamma)^{\myop} \ni (\sigma : \Delta \to \Gamma)
  \mapsto \mycat{E}(\Delta \vdash A\sigma) \in \mySet\]
is representable, where \(A\sigma\) denotes the element
\(P(\sigma)(A) \in P(\Delta)\) for a presheaf \(P\), a morphism
\(\sigma : \Delta \to \Gamma\) and an element \(A \in P(\Gamma)\).
We assume that any category with families
\(\mycat{E}\) has a choice of a representing object for this
presheaf denoted by \(\pi_{A} : \Gamma.A \to \Gamma\) and called the
\emph{context extension} of \(A\). We also require that, for every
context \(\Gamma \in \mycat{E}\), there exist types \(C_{0} \in
  \mycat{E}(\cdot), C_{1} \in \mycat{E}(\cdot.C_{0}), \dots, C_{n} \in
  \mycat{E}(\cdot.C_{0}.\dots.C_{n-1})\) and an isomorphism
\(\cdot.C_{0}.\dots.C_{n} \cong \Gamma\).
This means that, having dependent sum types, every context
\(\Gamma\) can be thought of a closed type \({} \vdash \Gamma\).
Type formers are modeled by algebraic operators. For example, to
model dependent product types, \(\mycat{E}\) has an operator \(\Pi\)
that carries triples \((\Gamma, A, B)\) consisting of a context
\(\Gamma\) and types \(A \in \mycat{E}(\Gamma)\) and \(B \in
  \mycat{E}(\Gamma.A)\) to types \(\Pi(\Gamma, A, B) \in
  \mycat{E}(\Gamma)\) and a bijection \(l(\Gamma, A, B) :
  \mycat{E}(\Gamma \vdash \Pi(\Gamma, A, B)) \cong \mycat{E}(\Gamma.A
  \vdash B)\). These operators must be stable under base changes, that
is, for any morphism \(\sigma : \Delta \to \Gamma\), we have
\(\Pi(\Gamma, A, B)\sigma = \Pi(\Delta, A\sigma, B\sigma)\) and
\(l(\Gamma, A, B)\sigma = l(\Delta, A\sigma, B\sigma)\). All
type-theoretic operations we introduce are required to be
stable under base changes, unless otherwise stated.
Note that there are alternative choices of notions of model of
dependent type theory including categories with attributes
\cite{cartmell1978generalised} and split full comprehension categories
\cite{jacobs1993comprehension}. Whichever model is
chosen, we proceed entirely in its internal language.

In dependent type theory, a type \(\Gamma \vdash \varphi\) is said
to be a \emph{proposition}, written \(\Gamma \vdash \varphi \ \myProp\),
if \(\Gamma, u_{1}, u_{2} : \varphi \vdash u_{1} = u_{2}\)
holds. For a proposition \(\Gamma \vdash \varphi\), we say
\(\varphi\) \emph{holds} if there exists a (unique) inhabitant of
\(\varphi\). For a type \(\Gamma \vdash A\), its \emph{propositional
truncation} \cite{awodey2004propositions} is a proposition \(\Gamma \vdash \|A\|\) equipped with a
constructor \(\Gamma, a : A \vdash |a| : \|A\|\) such that, for
every proposition \(\Gamma \vdash \varphi\), the function \(\Gamma
  \vdash \lambda fa.f(|a|) : (\|A\| \to \varphi) \to (A \to \varphi)\)
is an isomorphism. Propositions are closed under empty type,
cartesian products and dependent products along arbitrary types, and
we write \(\bot, \top, \varphi \land \psi, \forall_{x :
  A}\varphi(x)\) for \(\myEmptytype, \myUnittype, \varphi \times \psi,
  \prod_{x : A}\varphi(x)\), respectively, when emphasizing that they
are propositions. Also the identity type \(\myId(A, a_{0}, a_{1})\)
is a proposition because it is extensional, and often written
\(a_{0} = a_{1}\). The other logical operators are defined using
propositional truncation as \(\varphi \lor \psi := \|\varphi +
  \psi\|\) and \(\exists_{x : A}\varphi(x) := \|\sum_{x :
  A}\psi(x)\|\). One can show that these logical operations satisfy
the derivation rules of first-order intuitionistic logic. Moreover,
the type theory admits subset comprehension defined as
\[\Gamma \vdash \{x : A \mid \varphi(x)\} := \sum_{x : A}\varphi(x)\]
for a proposition \(\Gamma, x : A \vdash \varphi(x)\).

A finite coproduct \(A + B\) is said to be \emph{disjoint} if the
inclusions \(\myinl : A \to A + B\) and \(\myinr : B \to A + B\) are
monic and \(\forall_{a : A}\forall_{b : B}\myinl(a) \neq \myinr(b)\)
holds. A proposition \(\Gamma \vdash \varphi\) is said to be
\emph{decidable} if \(\Gamma \vdash \varphi \lor \neg \varphi\) holds. If the
coproduct \(\myBoolean := \myUnittype + \myUnittype\) of two copies
of the unit type is disjoint, then it is a \emph{decidable subobject
classifier}: for every decidable proposition \(\Gamma \vdash \varphi\),
there exists a unique term \(\Gamma \vdash b : \myBoolean\) such that
\(\Gamma \vdash \varphi \myfromto (b = 1)\)
holds. For readability we identify a boolean value \(b :
  \myBoolean\) with the proposition \(b = 1\).

For a functor \(H : \mycat{E} \to \mycat{F}\) between the underlying
categories of categories with families \(\mycat{E}\) and
\(\mycat{F}\), a \emph{dependent right adjoint} \cite{clouston2018modal} to
\(H\) consists of, for each context \(\Gamma \in \mycat{E}\) and
type \(A \in \mycat{F}(H\Gamma)\), a type \(G_{\Gamma}A \in
  \mycat{E}(\Gamma)\) and an isomorphism \(\varphi_{A} :
  \mycat{F}(H\Gamma \vdash A) \cong \mycat{E}(\Gamma \vdash
  G_{\Gamma}A)\) that are stable under reindexing in the sense that,
for any morphism \(\sigma : \Delta \to \Gamma\), we have
\((G_{\Gamma}A)\sigma = G_{\Delta}(A\sigma)\) and
\((\varphi_{A}a)\sigma = \varphi_{A\sigma}(a\sigma)\) for any \(a
  \in \mycat{F}(H\Gamma \vdash A)\). One can show that \(H\) preserves
all colimits whenever it has a dependent right adjoint.  As a
consequence, assuming the exponential functor \((\myI \to -)\) has a
dependent right adjoint, the interval \(\myI\) is connected
\[\forall_{\varphi : \myI \to \myBoolean}(\forall_{i : \myI}\varphi
  i) \lor (\forall_{i : \myI}\neg \varphi i),\] which is postulated in
\cite{orton2016axioms} as an axiom.

A \emph{universe} (\`{a} la Tarski) is a type \({} \vdash U\) equipped with
a type \(U \vdash \myel_{U}\). We often omit the subscript
\({}_{U}\) and simply write \(\myel\) for \(\myel_{U}\) if the
universe is clear from the context. The universe \(U\) is said to be
\emph{propositional} if \(U \vdash \myel_{U}\) is a proposition. An
\emph{impredicative universe} is a universe \(U\) equipped with the following
operations.
\begin{itemize}
\item A term \(A : U, B : \myel(A) \to U \vdash \sum^{U}(A, B) : U\)
equipped with an isomorphism \(A : U, B : \myel(A) \to U \vdash
    e : \myel(\sum^{U}(A, B)) \cong \sum_{x : \myel(A)}\myel(Bx)\).
\item A term \(A : U, a_{0}, a_{1} : \myel(A) \vdash \myId^{U}(A,
    a_{0}, a_{1}) : U\) equipped with an isomorphism \(A : U, a_{0},
    a_{1} : \myel(A) \vdash e : \myel(\myId^{U}(A, a_{0}, a_{1}))
    \cong (a_{0} = a_{1})\).
\item For every type \(\Gamma \vdash A\), a term \(\Gamma, B :
    \myel(A) \to U \vdash \prod^{U}(A, B) : U\) equipped with an
isomorphism \(\Gamma, B : \myel(A) \to U \vdash e :
    \myel(\prod^{U}(A, B)) \cong \prod_{x : A}\myel(Bx)\).
\end{itemize}
One might want to require that \(\myel(\sum^{U}(A, B))\) is equal to
\(\sum_{x : \myel(A)}\myel(Bx)\) on the nose rather than up to
isomorphism, but in the category of assemblies described in
\cref{sec:orga3e7717}, the impredicative universe of partial
equivalence relations does not satisfy this equation. For this
reason, the distinction between terms \(A : U\) and types
\(\myel(A)\) is necessary, but for readability we often identify
a term \(A : U\) with the type \(\myel(A)\).
For example, in \cref{org0fe33f4} some
\(\myel\)'s should be inserted formally. Also \cref{org629a444}
formally means that there exists a term \(\varphi, \psi : \myCof
  \vdash \lor^{\myCof}(\varphi, \psi) : \myCof\) such that \(\varphi,
  \psi : \myCof \vdash \myel(\lor^{\myCof}(\varphi, \psi)) \myfromto
  (\myel(\varphi) \lor \myel(\psi))\) holds.

Almost all the axioms in \cref{orgda6ce3b} are direct
translations of those in \cite{orton2016axioms,orton2017axioms}.
Strictly speaking,
\cref{orge92c1ca,orgc87da6c,org629a444,org8e1319f,org4e8c5eb} are part
of
structures rather than axioms in our setting, because \(\myCof\) is
no longer a subobject of the subobject classifier. Also
\cref{org0fe33f4}, called the \emph{isomorphism extension axiom}, is
part of structures. As already mentioned, the connectedness of the
interval \(\myI\) follows from the existence of the right adjoint to
the exponential functor \((\myI \to -)\). We need
\cref{org46f619f}, which asserts the extensionality of the
propositional universe \(\myCof\), for fibration structures on
identity types. This axiom trivially holds in case that \(\myCof\)
is a subobject of the subobject classifier in an elementary
topos. We also note that \(\myCof\) is closed under \(\bot\),
\(\top\) and \(\land\) using \cref{org09c88da,orgc87da6c,org8e1319f}.

\section{Modeling Cubical Type Theory}
\label{sec:orgb4d8a75}
We describe how to construct a model of a variant of cubical type
theory in our setting following Orton and Pitts
\cite{orton2016axioms}. Throughout the section \(\mycat{E}\) will be a
model of dependent type theory satisfying the conditions explained
in \cref{sec:org52feb89}. Type-theoretic notations in this
section are understood in the internal language of \(\mycat{E}\).

Cubical type theory is an extension of dependent type theory with an
\emph{interval object} {\cite[Section 3]{cohen2016cubical}}, the \emph{face lattice} {\cite[Section 4.1]{cohen2016cubical}},
\emph{systems} {\cite[Section 4.2]{cohen2016cubical}}, \emph{composition operations} {\cite[Section 4.3]{cohen2016cubical}} and the
\emph{gluing operation} {\cite[Section 6]{cohen2016cubical}}. It also has several type formers
including dependent product types, dependent sum types, \emph{path types}
{\cite[Section 3]{cohen2016cubical}} and, optionally, \emph{identity types} {\cite[Section 9.1]{cohen2016cubical}}. We make
some modifications to the original cubical type theory
\cite{cohen2016cubical} in the same way as Orton and Pitts
\cite{orton2016axioms}. Major differences are as follows.
\begin{enumerate}
\item In \cite{cohen2016cubical} the interval object \(\myI\) is a de
Morgan algebra, while we only require that \(\myI\) is a path
connection algebra.
\item Due to the lack of de Morgan involution, we need composition
operations in both directions ``from \(0\) to \(1\)'' and
``from \(1\) to \(0\)''.
\end{enumerate}
In this section we will construct from \(\mycat{E}\) a new model of
dependent type theory \(\mycat{E}^{F}\) that supports all operations
of cubical type theory.

\subsection{The Face Lattice and Systems}
\label{sec:orgd0c7e84}

The \emph{face lattice} {\cite[Section 4.1]{cohen2016cubical}} is modeled by the propositional
universe \(\myCof\). Note that in \cite{cohen2016cubical}
quantification \(\forall_{i : \myI}\varphi\) is not part of
syntax and written as a disjunction of irreducible elements, and
plays a crucial role for defining composition operation for gluing.
Since \(\myCof\) need not admit quantifier
elimination, we explicitly require \cref{org4e8c5eb}.

We use the following operation for modeling \emph{systems} {\cite[Section 4.2]{cohen2016cubical}}
which allows one to amalgamate compatible partial functions.

\begin{proposition}
\label{org77abe69}
One can derive an operation
\begin{mathpar}
  \inferrule
  {\Gamma \vdash A \\
    \Gamma \vdash \varphi_{i} \ \myProp \\
    \Gamma, u_{i} : \varphi_{i} \vdash a_{i}(u_{i}) : A \\
    \Gamma, u : \varphi_{i}, u' : \varphi_{j} \vdash a_{i}(u) = a_{j}(u') \\
    \text{(\(i\) and \(j\) run over \(\{1, \dots, n\}\))}}
  {\Gamma \vdash [(u_{1} : \varphi_{1}) \mapsto a_{1}(u_{1}), \dots, (u_{n} : \varphi_{n}) \mapsto a_{n}(u_{n})] : \varphi_{1} \lor \dots \lor \varphi_{n} \to A}
\end{mathpar}
such that \(\Gamma, v : \varphi_{i} \vdash [(u_{1} : \varphi_{1})
   \mapsto a_{1}(u_{1}), \dots, (u_{n} : \varphi_{n}) \mapsto
   a_{n}(u_{n})]v = a_{i}(v)\) for \(i = 1, \dots, n\).
\end{proposition}

\begin{proof}
Let \(B\) denote the union of images of \(a_{i}\)'s:
\[\Gamma \vdash B := \{a : A \mid (\exists_{u_{1} :
   \varphi_{1}}a_{1}(u_{1}) = a) \lor \dots \lor (\exists_{u_{n} :
   \varphi_{n}}a_{n}(u_{n}) = a)\}.\]
Then \(\Gamma \vdash B\) is a proposition because \(\Gamma, u :
   \varphi_{i}, u' : \varphi_{j} \vdash a_{i}(u) = a_{j}(u')\) for all
\(i\) and \(j\). Hence the function \([a_{1}, \dots, a_{n}] :
   \varphi_{1} + \dots + \varphi_{n} \to B\) induces a function
\(\|\varphi_{1} + \dots + \varphi_{n}\| \to B\).
\end{proof}

\subsection{Fibrations}
\label{sec:orgd04bdaa}
We regard the type of Boolean values \(\myBoolean\) as a subtype of
the interval \(\myI\) via the end-point inclusion \([0, 1] :
   \myBoolean \cong \myUnittype + \myUnittype \to \myI\). We define a
term \(e : \myBoolean \vdash \bar{e} : \myBoolean\) as \(\bar{0}
   = 1\) and \(\bar{1} = 0\).

\begin{definition}
For a type \(\Gamma, i : \myI \vdash A(i)\), we define a type of
\emph{composition structures} as
\[\Gamma \vdash \myComp^{i}(A(i)) := {} \prod_{e :
   \myBoolean}\prod_{\varphi : \myCof}\prod_{f : \varphi \to
   \prod_{i : \myI}A(i)}\prod_{a : A(e)} (\forall_{u : \varphi}fue =
   a) \to \{a' : A(\bar{e}) \mid \forall_{u : \varphi}fu\bar{e} =
   a'\}.\]
In this notation, the variable \(i\) is considered to be bound.
\end{definition}

\begin{definition}
For a type \(\gamma : \Gamma \vdash A(\gamma)\), we define a type of \emph{fibration
structures} as
\[{} \vdash \myFib(A) := \prod_{p : \myI \to
   \Gamma}\myComp^{i}(A(pi)).\]
A \emph{fibration} is a type \(\Gamma \vdash A\) equipped with a global
section \({} \vdash \alpha : \myFib(A)\).
\end{definition}

For a fibration structure \(\alpha : \myFib(A)\) on a type
\(\gamma : \Gamma \vdash A(\gamma)\) and a morphism \(\sigma :
   \Delta \to \Gamma\), we define
a fibration structure \(\alpha\sigma : \myFib(A\sigma)\) on
\(\delta : \Delta \vdash A(\sigma(\delta))\) as
\[\alpha\sigma = \lambda p.\alpha(\sigma \circ p) : \prod_{p :
   \myI \to \Delta}\myComp^{i}(A(\sigma(pi))).\]
Thus, for a fibration \((A, \alpha)\) on \(\Gamma\), we have its
\emph{base change} \((A\sigma, \alpha\sigma)\) along a morphism
\(\sigma : \Delta \to \Gamma\). With this base change operation
we get a model \(\mycat{E}^{F}\) of dependent type theory where
\begin{itemize}
\item the contexts are those of \(\mycat{E}\);
\item the types over \(\Gamma\) are fibrations over \(\Gamma\);
\item the terms of a fibration \(\Gamma \vdash A\) are terms of the
underlying type \(\Gamma \vdash A\) in \(\mycat{E}\)
\end{itemize}
together with a forgetful map \(\mycat{E}^{F} \to \mycat{E}\).
In the same way as Orton and Pitts \cite{orton2016axioms}, one can
show the following.

\begin{theorem}
\label{org6d3dfea}
The model of dependent type theory \(\mycat{E}^{F}\) supports:
\begin{itemize}
\item composition operations, path types and identity types; and
\item dependent product types, dependent sum types, unit type and
finite coproducts preserved by the forgetful map \(\mycat{E}^{F}
     \to \mycat{E}\).
\end{itemize}
\end{theorem}

We also introduce a class of objects that automatically carry
fibration structures.

\begin{definition}
A type \({} \vdash A\) is said to be \emph{discrete} if \(\forall_{f :
   \myI \to A}\forall_{i : \myI}fi = f0\) holds.
\end{definition}

\begin{proposition}
\label{org96b19e9}
If \({} \vdash A\) is a discrete type, then it has a fibration
structure.
\end{proposition}

\begin{proof}
Let \(e : \myBoolean\), \(\varphi : \myCof\), \(f : \varphi
   \to \myinterval \to A\) and \(a :  A\) such that \(\forall_{u :
   \varphi}fue = a\). Then \(a' := a : A\) satisfies \(\forall_{u :
   \varphi}fu\bar{e} = a'\) by the discreteness.
\end{proof}

\subsection{Path Types and Identity Types}
\label{sec:orga62704e}

For a type \(\Gamma \vdash A\) and terms \(\Gamma \vdash a_{0} :
   A\) and \(\Gamma \vdash a_{1} : A\), we define the \emph{path type}
\(\Gamma \vdash \myPath(A, a_{0}, a_{1})\) to be
\[\Gamma \vdash \{p : \myI \to A \mid p0 = a_{0} \land p1 =
   a_{1}\}.\]
We also define the \emph{identity type} \(\Gamma \vdash \myId(A,
   a_{0}, a_{1})\) to be
\[\Gamma \vdash \sum_{p : \myPath(A, a_{0}, a_{1})}\{\varphi :
   \myCof \mid \varphi \to \forall_{i : \myI}pi = a_{0}\}\]
which is a variant of Swan's construction
\cite{swan2016algebraic}. \Cref{org6d3dfea} says that, if
\(A\) has a fibration structure, then so do \(\myPath(A, a_{0},
   a_{1})\) and \(\myId(A, a_{0}, a_{1})\).

In the model \(\mycat{E}^{F}\), both path types and identity types
admit the following introduction and elimination operations:
\begin{mathpar}
  \inferrule
  {\Gamma \vdash a : A}
  {\Gamma \vdash \myrefl_{a} : P(A, a, a)}
  \quad \text{\(P\)-intro}
  \and
  \inferrule
  {\Gamma, x_{0} : A, x_{1} : A, z : P(A, x_{0}, x_{1}) \vdash C(z) \\
    \Gamma, x : A \vdash c(x) : C(\myrefl_{x}) \\
    \Gamma \vdash a_{0} : A \\
    \Gamma \vdash a_{1} : A \\
    \Gamma \vdash p : P(A, a_{0}, a_{1})}
  {\Gamma \vdash \myind_{P(A)}(C, c, p) : C(p)}
  \quad \text{\(P\)-elim}
\end{mathpar}
where \(P\) is either \(\myPath\) or \(\myId\). A difference
between them is their computation rules. Identity types admit the
\emph{judgmental} computation rule like Martin-L\"{o}f's identity
types:
\[\Gamma \vdash \myind_{\myId(A)}(C, c, \myrefl_{a}) = c(a)\]
for a term \(\Gamma \vdash a : A\). On the other hand, path types
only admit the \emph{propositional} computation rule: for a term
\(\Gamma \vdash a : A\), one can find a term
\[\Gamma \vdash H(C, c, a) : \myPath(C(a), \myind_{\myPath(A)}(C,
   c, \myrefl_{a}), c(a)).\]
Therefore, when interpreting homotopy type theory, which is based
on Martin-L\"{o}f's type theory, we use \(\myId(A, a_{0},
   a_{1})\) rather than \(\myPath(A, a_{0}, a_{1})\). However, it can
be shown that \(\myId(A, a_{0}, a_{1})\) and \(\myPath(A, a_{0},
   a_{1})\) are equivalent, and thus we can
replace \(\myId(A, a_{0}, a_{1})\) by simpler type \(\myPath(A,
   a_{0}, a_{1})\) when analyzing the model \(\mycat{E}^{F}\) (see,
for instance, the definition of homotopy proposition in
\cref{sec:org7cd507b}).

\subsection{Universes and Gluing}
\label{sec:org7bb1bad}
For a type \(\gamma : \Gamma \vdash A(\gamma)\), a fibration structure on \(A\)
corresponds to a term of the type \(p : \myI \to \Gamma \vdash
   \myCComp(A)(p) := \myComp^{i}(A(pi))\). We define a type \(\Gamma
   \vdash FA := \myCComp(A)_{\myI}\), using the dependent right
adjoint \((-)_{\myI}\) to the exponential functor \((\myI \to -)\).
By definition a morphism
\(\sigma : \Delta \to \sum_{\Gamma}FA\) corresponds to a pair
\((\sigma_{0}, \alpha)\) consisting of a morphism \(\sigma_{0} :
   \Delta \to \Gamma\) and a fibration structure \({} \vdash \alpha :
   \prod_{p : \myI \to \Delta}\myComp^{i}(A(\sigma_{0}(pi)))\).

Using this construction for the universe \(\myUniverse \vdash
   \myel_{\myUniverse}\), we have a new universe \(\myUniverse^{F} :=
   \sum_{\myUniverse}F(\myel)\) together with a fibration \((A,
   \alpha) : \myUniverse^{F} \vdash \myel_{\myUniverse^{F}}(A, \alpha)
   := \myel_{\myUniverse}(A)\). By definition \(\myUniverse^{F}\)
classifies fibrations whose underlying types belong to
\(\myUniverse\).

\begin{theorem}
\label{orgbefc7b4}
The universe \(\myUniverse^{F}\) is closed under dependent product
types along arbitrary fibrations, dependent sum types and path
types. If \(\myCof\) belongs to \(\myUniverse\), then
\(\myUniverse^{F}\) is closed under identity types.
\end{theorem}

\begin{proof}
By \cref{org6d3dfea}, it suffices to show
that \(\myUniverse\) is closed under those type
constructors, but this is clear by definition.
\end{proof}

We describe the \emph{gluing operation} on the universe
\(\myUniverse^{F}\) following Orton and Pitts
\cite{orton2016axioms}.

For a proposition \(\Gamma \vdash \varphi\), types \(\Gamma, u :
   \varphi \vdash A(u)\) and \(\Gamma \vdash B\) and a function
\(\Gamma, u : \varphi \vdash f(u) : A(u) \to B\), we define a type
\(\myGlue(\varphi, f)\) to be
\[\Gamma \vdash \myGlue(\varphi, f) := \sum_{a : \prod_{u :
   \varphi}A(u)}\{b : B \mid \forall_{u : \varphi}f(u)(au) = b\}.\]
There is a canonical isomorphism \(\Gamma, u : \varphi \vdash
   e(u) := \lambda(a, b).au : \myGlue(\varphi, f) \cong A(u)\)
with inverse \(\lambda a.(\lambda v.a, f(u)a)\).

\begin{proposition}
\label{orgeaa893e}
For \(\gamma : \Gamma \vdash \varphi(\gamma) : \myCof\), \(\gamma :
   \Gamma, u : \varphi(\gamma)
   \vdash A(u)\), \(\gamma : \Gamma \vdash B(\gamma)\) and \(\gamma : \Gamma, u : \varphi(\gamma) \vdash
   f(u) : A(u) \to B\), if \(A\) and \(B\) are fibrations and \(f\) is
an equivalence, then \(\gamma : \Gamma \vdash \myGlue(\varphi(\gamma), f)\) has a
fibration structure preserved by the canonical isomorphism
\(\Gamma, u : \varphi \vdash e(u) : \myGlue(\varphi, f)
   \cong A(u)\).
\end{proposition}

\begin{proof}
The construction is similar to the definition of the composition
operation for glue types {\cite[Section 6.2]{cohen2016cubical}}.
\end{proof}

Since the universe \(\myUniverse\) is closed under type formers
used in the definition of \(\myGlue(\varphi, f)\), we get a term
\[\varphi : \myCof, A : \varphi \to \myUniverse, B : \myUniverse,
   f : \prod_{u : \varphi}A(u) \to B \vdash \myGlue(\varphi, f) :
   \myUniverse\]
such that \(\prod_{u : \varphi}\myGlue(\varphi, f) \cong
   A(u)\). However, the gluing operation in cubical type theory
{\cite[Section 6]{cohen2016cubical}} requires that, assuming \(u : \varphi\),
\(\myGlue(\varphi, f)\) is equal to \(A(u)\) on the nose rather
than up to isomorphism. So we use \cref{org0fe33f4} and
get a term
\[\varphi : \myCof, A : \varphi \to \myUniverse, B : \myUniverse,
   f : \prod_{u : \varphi}A(u) \to B \vdash \mySGlue(\varphi, f) :
   \myUniverse\]
such that \(\mySGlue(\varphi, f) \cong \myGlue(\varphi, f)\) and
\(\forall_{u : \varphi}\mySGlue(\varphi, f) = A(u)\). By
\cref{orgeaa893e} we also have a term
\[\varphi : \myCof, A : \varphi \to \myUniverse^{F}, B :
   \myUniverse^{F}, f : \prod_{u : \varphi}A(u) \simeq B \vdash
   \mySGlue(\varphi, f) : \myUniverse^{F}\]
such that \(\mySGlue(\varphi, f) \cong \myGlue(\varphi, f)\) and
\(\forall_{u : \varphi}\mySGlue(\varphi, f) = A(u)\). Hence the
universe \(\myUniverse^{F}\) in the model \(\mycat{E}^{F}\)
supports the gluing operation. The composition operation for
universes is defined using the gluing operation {\cite[Section 7.1]{cohen2016cubical}}, so we
have the following proposition.

\begin{proposition}
\label{org97bc451}
\({} \vdash \myUniverse^{F}\) has a fibration structure.
\end{proposition}

Since the univalence axiom can be derived from
the gluing operation {\cite[Section 7]{cohen2016cubical}}, we conclude that
\(\myUniverse^{F}\) is a univalent and impredicative universe in
the model of cubical type theory \(\mycat{E}^{F}\).

\section{Presheaf Models}
\label{sec:org3425270}
In this section we give a sufficient condition for a presheaf
category to satisfy the conditions in
\cref{sec:org52feb89}. We will work in a model \(\mycat{S}\) of
dependent type theory with dependent product types, dependent sum
types, extensional identity types, unit type, disjoint finite
coproducts and propositional truncation.

A \emph{category} in \(\mycat{S}\) consists of:
\begin{itemize}
\item a type \({} \vdash \myicat{C}_{0}\) of \emph{objects};
\item a type \(c_{0}, c_{1} : \myicat{C}_{0} \vdash
    \myicat{C}_{1}(c_{0}, c_{1})\) of \emph{morphisms};
\item a term \(c : \myicat{C}_{0} \vdash \myid_{c} : \myicat{C}_{1}(c,
    c)\) called \emph{identity};
\item a term \(c_{0}, c_{1}, c_{2} : \myicat{C}_{0}, g :
    \myicat{C}_{1}(c_{1}, c_{2}), f : \myicat{C}_{1}(c_{0}, c_{1})
    \vdash gf : \myicat{C}_{1}(c_{0}, c_{2})\) called
\emph{composition}
\end{itemize}
satisfying the standard axioms of category. We will simply write
\(\myicat{C}\) and \(\myicat{C}(c_{0}, c_{1})\) for
\(\myicat{C}_{0}\) and \(\myicat{C}_{1}(c_{0}, c_{1})\)
respectively.  The notions of functor and natural transformation
in \(\mycat{S}\) are defined in the obvious way. For a category
\(\myicat{C}\) in \(\mycat{S}\), a \emph{presheaf} on \(\myicat{C}\)
consists of:
\begin{itemize}
\item a type \(c : \myicat{C} \vdash A(c)\);
\item a term \(c_{0}, c_{1} : \myicat{C}, \sigma : \myicat{C}(c_{0},
    c_{1}), a : A(c_{1}) \vdash a\sigma : A(c_{0})\) called \emph{(right)
\(\myicat{C}\)-action}
\end{itemize}
satisfying \(a\myid = a\) and \(a(\sigma\tau) = (a\sigma)\tau\). For
presheaves \(A\) and \(B\), a \emph{morphism} \(f : A \to B\) is a term
\(c : \myicat{C}, a : A(c) \vdash f(a) : B(c)\) satisfying \(c_{0},
  c_{1} : \myicat{C}, \sigma : \myicat{C}(c_{0}, c_{1}), a : A(c_{1})
  \vdash f(a\sigma) = f(a)\sigma\).  For a presheaf \(A\), its
\emph{category of elements}, written \(\mycatEl(A)\), is defined
as
\begin{itemize}
\item \({} \vdash \mycatEl(A)_{0} := \sum_{c : \myicat{C}_{0}}A(c)\);
\item \((c_{0}, a_{0}), (c_{1}, a_{1}) : \mycatEl(A)_{0} \vdash
    \mycatEl(A)_{1}((c_{0}, a_{0}), (c_{1}, a_{1})) :=
    \{\sigma : \myicat{C}_{1}(c_{0}, c_{1}) \mid a_{1}\sigma =
    a_{0}\}\).
\end{itemize}
There is a projection functor \(\pi_{A} : \mycatEl(A) \to
  \myicat{C}\).

For a category \(\myicat{C}\) in \(\mycat{S}\), we describe the
presheaf model \(\myPrshv(\myicat{C})\) of dependent type
theory. Contexts are interpreted as
presheaves on \(\myicat{C}\). For a context \(\Gamma\), types on
\(\Gamma\) are interpreted as presheaves on
\(\mycatEl(\Gamma)\). For a type \(\Gamma \vdash A\), terms
of \(A\) are interpreted as sections of the projection \(\pi_{A} :
  \mycatEl(A) \to \mycatEl(\Gamma)\). For a
type \(\Gamma \vdash A\), the context extension \(\Gamma.A\) is
interpreted as the presheaf \(c : \myicat{C} \vdash \sum_{\gamma :
  \Gamma(c)}A(c, \gamma)\). This construction is also used for
dependent sum types. The dependent product for a type \(\Gamma.A \vdash B\) is the
presheaf
\begin{align*}
  (c, \gamma) : \mycatEl(\Gamma) \vdash {}
  &\{f : \prod_{c' : \myicat{C}}\prod_{\sigma : \myicat{C}(c', c)}\prod_{a : A(c', \gamma\sigma)}B(c', a) \mid \\
  &\forall_{c', c'' : \myicat{C}}\forall_{\sigma : \myicat{C}(c', c)}\forall_{\tau : \myicat{C}(c'', c')}\forall_{a : A(c', \gamma\sigma)}(fc'\sigma a)\tau = fc''(\sigma\tau)(a\tau)\}.
\end{align*}
Extensional identity types, unit type, disjoint finite coproducts
and propositional truncation are pointwise.

\subsection{Lifting Universes}
\label{sec:org2c86a34}

We describe the Hofmann-Streicher lifting of a universe
\cite{hofmann1997lifting}. Let \(\myicat{C}\) be a category in
\(\mycat{S}\) and \(U\) a universe in \(\mycat{S}\). We define a
universe \([\myicat{C}^{\myop}, U]\) in \(\myPrshv(\myicat{C})\) as
follows. The universe \(U\) can be seen as a category whose type of
objects is \(U\) and type of morphisms is \(A, B : U \vdash
   \myel_{U}(A) \to \myel_{U}(B)\). For an object \(c : \myicat{C}\),
we define \([\myicat{C}^{\myop}, U](c)\) to be the type of functors
from \((\myicat{C}/c)^{\myop}\) to \(U\). The \(\myicat{C}\)-action
on \([\myicat{C}^{\myop}, U]\) is given by precomposition. The type
\([\myicat{C}^{\myop}, U] \vdash \myel_{[\myicat{C}^{\myop}, U]}\)
in \(\myPrshv(\myicat{C})\) is defined as \((c, A) :
   \mycatEl([\myicat{C}^{\myop}, U]) \vdash
   \myel_{[\myicat{C}^{\myop}, U]}(c, A) := \myel_{U}(A(\myid_{c}))\).

It is easy to show that, if \(U\) is an impredicative universe,
then dependent product types, dependent sum types and extensional
identity types in \(U\) can be lifted to those in
\([\myicat{C}^{\myop}, U]\) so that \([\myicat{C}^{\myop}, U]\) is
an impredicative universe in \(\myPrshv(\myicat{C})\). If \(U\) is
a propositional universe in \(\mycat{S}\), then
\([\myicat{C}^{\myop}, U]\) is a propositional universe in
\(\myPrshv(\myicat{C})\).

\begin{proposition}
\label{org824da5b}
Let \(\myUniverse\) be an impredicative universe and \(\myCof\) a
propositional universe in \(\mycat{S}\). If they satisfy
\cref{org629a444,org8e1319f,org46f619f,org0fe33f4}, then so do
\([\myicat{C}^{\myop}, \myUniverse]\) and \([\myicat{C}^{\myop},
\myCof]\).
\end{proposition}

\begin{proof}
We only check \cref{org0fe33f4}. The other axioms are
easy to verify.

We have to define a term \(\varphi : [\myicat{C}^{\myop}, \myCof],
   A : \varphi \to [\myicat{C}^{\myop}, \myUniverse], B :
   [\myicat{C}^{\myop}, \myUniverse], f : \prod_{u : \varphi}Au \cong
   B \vdash (D(\varphi, f), g(\varphi, f)) : \sum_{\bar{A} :
   [\myicat{C}^{\myop}, \myUniverse]}\{\bar{f} : \bar{A} \cong B \mid
   \forall_{u : \varphi}(Au, fu) = (\bar{A}, \bar{f})\}\) in
\(\myPrshv(\myicat{C})\). It corresponds to a natural
transformation that takes an object \(c : \myicat{C}\), functors
\(\varphi : (\myicat{C}/c)^{\myop} \to \myCof\), \(A :
   \mycatEl(\varphi)^{\myop} \to \myUniverse\) and \(B :
   (\myicat{C}/c)^{\myop} \to \myUniverse\) and an isomorphism \(f : A
   \cong B\pi_{\varphi}\) of presheaves on
\(\mycatEl(\varphi)\) and returns a pair \((D(c, \varphi,
   f), g(c, \varphi, f))\) consisting of a functor \(D(c, \varphi, f) :
   (\myicat{C}/c)^{\myop} \to \myUniverse\) and an isomorphism \(g(c,
   \varphi, f) : A \cong B\) of presheaves on
\((\myicat{C}/c)^{\myop}\) such that \(D(c, \varphi,
   f)\pi_{\varphi} = A\) and \(g(c, \varphi, f)\pi_{\varphi} =
   f\). Let \(\sigma : \myicat{C}(c', c)\) be a morphism. Then we have
\(\varphi(\sigma) : \myCof\), \(\lambda u.A(\sigma, u) :
   \varphi(\sigma) \to \myUniverse\), \(B(\sigma) : \myUniverse\) and
an isomorphism \(\lambda u.f(\sigma, u) : \prod_{u :
   \varphi(\sigma)}A(\sigma, u) \cong B(\sigma)\). By the isomorphism
lifting on \(\myUniverse\), we have \(D(c, \varphi, f)(\sigma) :
   \myUniverse\) and an isomorphism \(g(c, \varphi, f)(\sigma) : D(c,
   \varphi, f)(\sigma) \cong B(\sigma)\) such that \(\forall_{u :
   \varphi(\sigma)}(A(\sigma, u), f(\sigma, u)) = (D(c, \varphi,
   f)(\sigma), g(c, \varphi, f)(\sigma))\). For the morphism part of
the functor \(D(c, \varphi, f)\), let \(\tau : \myicat{C}(c'',
   c')\) be another morphism. Then we define \(\tau^{*} : D(c,
   \varphi, f)(\sigma) \to D(c, \varphi, f)(\sigma\tau)\) to be the
composition
\[
  \begin{tikzcd}[column sep=11ex]
    D(c, \varphi, f)(\sigma) \arrow[r,"{g(c, \varphi, f)(\sigma)}","\cong"'] &
    B(\sigma) \arrow[r,"\tau^{*}"] &
    B(\sigma\tau) \arrow[r,"{g(c, \varphi, f)(\sigma\tau)^{-1}}","\cong"'] &
    D(c, \varphi, f)(\sigma\tau).
  \end{tikzcd}
\]
By definition \(g(c, \varphi, f)\) becomes a natural isomorphism and
\((D(c, \varphi, f)\pi_{\varphi}, g(c, \varphi, f)\pi_{\varphi}) = (A,
f)\). It is easy to see the naturality of
\((c, \varphi, f) \mapsto (D(c, \varphi, f), g(c, \varphi, f))\).
\end{proof}

\subsection{Intervals}
\label{sec:org5989482}
Suppose a category \(\myicat{C}\) in \(\mycat{S}\) has finite
products. A \emph{path connection algebra} in \(\myicat{C}\) consists of
an object \(\myI : \myicat{C}\), morphisms \(\delta_{0},
   \delta_{1} : \myicat{C}(1, \myI)\) called \emph{end-points} and
morphisms \(\mu_{0}, \mu_{1} : \myicat{C}(\myI \times \myI, \myI)\)
called \emph{connections} satisfying \(\mu_{e}(\delta_{e} \times \myI) =
   \mu_{e}(\myI \times \delta_{e}) = \delta_{e}\) and
\(\mu_{e}(\delta_{\bar{e}} \times \myI) = \mu_{e}(\myI \times
   \delta_{\bar{e}}) = \myid\) for \(e \in \{0, 1\}\).

For a path connection algebra \(\myI\) in \(\myicat{C}\), we have a
representable presheaf \(\myyoneda \myI\) on \(\myicat{C}\). Since
the Yoneda embedding is fully faithful and preserves finite
products, \(\myyoneda \myI\) has end-points and connections
satisfying \cref{org541f57d,org6fd8c58}. The
interval \(\myyoneda \myI\) satisfies \cref{org09c88da} if
and only if \(\forall_{c : \myicat{C}}\delta_{0}!_{c} \neq
   \delta_{1}!_{c}\) holds, where \(!_{c} : \myicat{C}(c, 1)\) is the
unique morphism into the terminal object.

\begin{proposition}
Let \(\myCof\) be a propositional universe in \(\mycat{S}\) and
suppose that, for every pair of objects \(c, c' : \myicat{C}\), the
equality predicate on \(\myicat{C}(c, c')\) belongs to
\(\myCof\). Then, for every object \(c : \myicat{C}\), the equality
predicate on \(\myyoneda c\) belongs to \([\myicat{C}^{\myop},
   \myCof]\). In particular, \(\myyoneda \myI\) and
\([\myicat{C}^{\myop}, \myCof]\) in \(\myPrshv(\myicat{C})\)
satisfy \cref{orge92c1ca,orgc87da6c}.
\end{proposition}

\begin{proof}
Because equality on a presheaf is pointwise.
\end{proof}

\begin{proposition}
For a functor \(f : \myicat{C} \to \myicat{D}\) between categories
in \(\mycat{S}\), the precomposition functor \(f^{*} :
   \myPrshv(\myicat{D}) \to \myPrshv(\myicat{C})\) has a dependent
right adjoint \(f_{*}\).
\end{proposition}

\begin{proof}
For a context \(\Gamma\) in \(\myPrshv(\myicat{D})\) and a type \(f^{*}\Gamma
   \vdash A\) in \(\myPrshv(\myicat{C})\), the type \(\Gamma \vdash
   f_{*}A\) is given by the presheaf \((d, \gamma) : \mycatEl(\Gamma)
   \vdash \lim_{(c, \sigma) : (f \downarrow d)}A(c, \gamma\sigma)\).
\end{proof}

\begin{proposition}
\label{org88bbd74}
Suppose that a category \(\myicat{C}\) in \(\mycat{S}\) has finite
products. For an object \(c : \myicat{C}\), the exponential functor
\((\myyoneda c \to -) : \myPrshv(\myicat{C}) \to
   \myPrshv(\myicat{C})\) is isomorphic to \((- \times c)^{*}\).
\end{proposition}

\begin{proof}
\((\myyoneda c \to A)(c') \cong \myPrshv(\myicat{C})(\myyoneda c'
   \times \myyoneda c, A) \cong \myPrshv(\myicat{C})(\myyoneda(c'
   \times c), A) \cong A(c' \times c)\).
\end{proof}

Hence the exponential functor \((\myyoneda \myI \to -)\) has a
dependent right adjoint. \Cref{org88bbd74} also
implies \cref{org4e8c5eb} for the propositional universe
\([\myicat{C}^{\myop}, \myCof]\). Explicitly, \(\forall_{\myyoneda
   \myI} : (- \times \myyoneda \myI)^{*}[\myicat{C}^{\myop}, \myCof]
   \to [\myicat{C}^{\myop}, \myCof]\) is a natural transformation that
carries a functor \(\varphi : (\myicat{C}/c \times \myI)^{\myop}
   \to \myCof\) to \(\lambda \sigma.\varphi(\sigma \times \myI) :
   (\myicat{C}/c)^{\myop} \to \myCof\).

In summary, we have:

\begin{theorem}
\label{org611c4ef}
Suppose:
\begin{itemize}
\item \(\mycat{S}\) is a model of dependent type theory with dependent
product types, dependent sum types, extensional identity types,
unit type, disjoint finite coproducts and propositional
truncation;
\item \(\myCof\) is a propositional universe and \(\myUniverse\) is an
impredicative universe satisfying
\cref{org629a444,org8e1319f,org46f619f,org0fe33f4};
\item \(\myicat{C}\) is a category in \(\mycat{S}\) with finite
products and the equality on \(\myicat{C}(c, c')\) belongs to
\(\myCof\) for every pair of objects \(c, c' : \myicat{C}\);
\item \(\myI\) is a path connection algebra in \(\myicat{C}\);
\item \(\myyoneda \myI\) satisfies \cref{org09c88da}.
\end{itemize}
Then the presheaf model \(\myPrshv(\myicat{C})\) together with
propositional universe \([\myicat{C}^{\myop}, \myCof]\), impredicative
universe \([\myicat{C}^{\myop}, \myUniverse]\) and interval
\(\myyoneda \myI\) satisfies all the axioms in
\cref{orgda6ce3b}.
\end{theorem}

\subsection{Decidable Subobject Classifier}
\label{sec:org1d99ffe}

An example of the propositional universe \(\myCof\) in
\cref{org611c4ef} is the decidable subobject classifier
\(\myBoolean\) which always satisfies
\cref{org629a444,org8e1319f,org46f619f}.

\begin{proposition}
In a model of dependent type theory with dependent product types,
dependent sum types, extensional identity types, unit type,
disjoint finite coproducts and propositional truncation, any
universe \(\myUniverse\) satisfies \cref{org0fe33f4} with
\(\myCof = \myBoolean\).
\end{proposition}

\begin{proof}
Let \(\varphi : \myBoolean, A : \varphi \to \myUniverse, B :
   \myUniverse, f : \prod_{u : \varphi}Au \cong B\). We define
\(\mylift(\varphi, f)\) by case analysis on \(\varphi :
   \myBoolean\) as \(\mylift(0, f) := (B, \myid)\) and \(\mylift(1, f)
   := (A{*}, f{*})\) where \({*}\) is the unique element of a
singleton type.
\end{proof}

\subsection{Categories of Cubes}
\label{sec:orgb20ba66}
We present examples of internal categories \(\myicat{C}\) with a path
connection algebra \(\myI\) satisfying the hypotheses of
\cref{org611c4ef} with \(\myCof = \myBoolean\). Obvious choices of
\(\myicat{C}\) are the category of free de Morgan algebras
\cite{cohen2016cubical} and various syntactic categories of the
language \(\{0, 1, \sqcap, \sqcup\}\) \cite{buchholtz2017varieties},
but some inductive types and quotient types are required to construct
these categories in dependent type theory. Although the motivating
example of \(\mycat{S}\), the category of assemblies described in
\cref{sec:orga3e7717}, has inductive types and finite colimits,
quotients are not well-behaved in general and we need to be careful in
using quotients. Instead, we give examples definable only using
natural numbers.

Suppose \(\mycat{S}\) is a model of dependent type theory with
dependent product types, dependent sum types, extensional identity
types, unit type, disjoint finite coproducts, propositional
truncation and natural numbers. We define a type of finite types
\(n : \myN \vdash \myFin_{n}\) to be \(\myFin_{n} = \{k : \myN \mid
   k < n\}\). We define a category \(\myicat{B}\) as follows. Its
object of objects is \(\myN\). The morphisms \(m \to n\) are
functions \((\myFin_{m} \to \myBoolean) \to (\myFin_{n} \to
   \myBoolean)\). In the category \(\myicat{B}\), the terminal object is
\(0 : \myN\) and the product of \(m\) and \(n\) is \(m + n\). One can
show, by induction, that every \(\myicat{B}(m, n)\) has decidable
equality. \(\myicat{B}\) has a path connection algebra \(1 : \myN\)
together with end-points \(0, 1 : (\myFin_{0} \to \myBoolean) \to
   (\myFin_{1} \to \myBoolean)\) and connections \(\mymin, \mymax :
   (\myFin_{1} \to \myBoolean) \times (\myFin_{1} \to \myBoolean) \to
   (\myFin_{1} \to \myBoolean)\). One can show that the category
\(\myicat{B}\) satisfies the hypotheses of
\cref{org611c4ef}. Moreover, any subcategory of \(\myicat{B}\)
that has the same finite products and contains the path connection
algebra \(1\) satisfies the same condition. An example is the wide
subcategory \(\myicat{B}_{\myord}\) of \(\myicat{B}\) where the
morphisms are order-preserving functions \((\myFin_{m} \to
   \myBoolean) \to (\myFin_{n} \to \myBoolean)\).

\subsection{Constant and Codiscrete Presheaves}
\label{sec:orgf2b0669}

We show some properties of constant and codiscrete presheaves which
will be used in \cref{sec:orga3e7717}.
Let \(\mycat{S}\) be a model of dependent type theory satisfying
the hypotheses of \cref{org611c4ef}. For an object \(A \in
   \mycat{S}\), we define the \emph{constant presheaf} \(\Delta A\) to be
\(\Delta A(c) := A\) with the trivial \(\myicat{C}\)-action.

\begin{proposition}
\label{orgf81c59c}
Every constant presheaf \(\Delta A\) is discrete.
\end{proposition}

\begin{proof}
For every \(c : \myicat{C}\), we have \((\myyoneda \myI \to \Delta A)(c)
   \cong \Delta A(c \times \myI) = A\) by
\cref{org88bbd74}.
\end{proof}

For a type \(\Gamma \vdash A\) in \(\mycat{S}\), we define the
\emph{codiscrete presheaf} \(\Delta \Gamma \vdash \nabla A\) to be
\(\nabla A(c, \gamma) := \myicat{C}(1, c) \to A(\gamma)\) with
composition as the \(\myicat{C}\)-action.

\begin{proposition}
\label{org242a9f3}
Suppose that \(\myCof = \myBoolean\). Then for every type \(\Gamma
   \vdash A\) in \(\mycat{S}\), the type \(\Delta \Gamma \vdash \nabla
   A\) has a fibration structure.
\end{proposition}

\begin{proof}
Since \(\Delta \Gamma\) is discrete, it suffices to show that
\(\nabla A(\gamma)\) has a fibration structure for every \(\gamma :
   \Gamma\). Thus we may assume that
\(\Gamma\) is the empty context. We construct a term
\[\alpha : \prod_{e : \myBoolean}\prod_{\varphi :
   [\myicat{C}^{\myop}, \myBoolean]}\prod_{f : \varphi \to \myI \to
   \nabla A}\prod_{a : \nabla A}(\forall_{u : \varphi}fue = a) \to
   \{\bar{a} : \nabla A \mid \forall_{u : \varphi}fu\bar{e} =
   \bar{a}\}\]
in \(\myPrshv(\myicat{C})\). It corresponds to a natural
transformation that takes an object \(c : \myicat{C}\), an element
\(e : \myBoolean\), a functor \(\varphi : (\myicat{C}/c)^{\myop}
   \to \myBoolean\), a natural transformation \(f :
   \int_{c' \in \myicat{C}}(\sum_{\sigma : \myicat{C}(c', c)}\varphi(\sigma)) \times
   \myicat{C}(c', \myI) \to \nabla A(c')\) and an element \(a : \nabla
   A(c)\) such that \(\forall_{c' : \myicat{C}}\forall_{\sigma :
   \myicat{C}(c', c)}\forall_{u : \varphi(\sigma)}f(\sigma, u, e) =
   a\sigma\) and returns an element \(\alpha(e, \varphi, f, a) :
   \nabla A(c)\) such that \(\forall_{c' :
   \myicat{C}}\forall_{\sigma : \myicat{C}(c', c)}\forall_{u :
   \varphi(\sigma)}f(\sigma, u, \bar{e}) = \alpha(e, \varphi, f,
   a)\sigma\). We define \(\alpha(e, \varphi, f, a) : \myicat{C}(1, c)
   \to A\) as
\[
  \alpha(e, \varphi, f, a)(\sigma) :=
  \left\{
    \begin{array}{ll}
      f(\sigma, u, \bar{e})(\myid_{1}) & \text{if \(u : \varphi(\sigma)\) is found} \\
      a(\sigma) & \text{otherwise}
    \end{array}
  \right.
\]
for \(\sigma : \myicat{C}(1, c)\). Then by definition
\(\forall_{c' : \myicat{C}}\forall_{\sigma : \myicat{C}(c',
   c)}\forall_{u : \varphi(\sigma)}f(\sigma, u, \bar{e}) = \alpha(e,
   \varphi, f, a)\sigma\).
\end{proof}

\begin{proposition}
\label{org7b84af3}
Suppose that \(\myicat{C}(1, \myI)\) only contains \(0\) and \(1\),
namely \(\forall_{\sigma : \myicat{C}(1, \myI)}\sigma = 0 \lor
   \sigma = 1\). Then for every type \(\Gamma \vdash A\) in
\(\mycat{S}\), there exists a term
\[\Delta \Gamma \vdash p : \prod_{a_{0}, a_{1} : \nabla
   A}\myPath(\nabla A, a_{0}, a_{1})\]
in \(\myPrshv(\myicat{C})\).
\end{proposition}

\begin{proof}
We may assume that \(\Gamma\) is the empty context. The term \(p\)
corresponds to a natural transformation that takes an object \(c :
   \myicat{C}\), elements \(a_{0}, a_{1} : \nabla A(c)\) and a
morphism \(i : \myicat{C}(c, \myI)\) and returns an element
\(p(a_{0}, a_{1}, i) : \nabla A(c)\) such that \(p(a_{0}, a_{1}, 0)
   = a_{0}\) and \(p(a_{0}, a_{1}, 1) = a_{1}\). We define \(p(a_{0},
   a_{1}, i) : \myicat{C}(1, c) \to A\) as
\[
  p(a_{0}, a_{1}, i)(\sigma) :=
  \left\{
    \begin{array}{ll}
      a_{0}(\sigma) & \text{if \(i\sigma = 0\)} \\
      a_{1}(\sigma) & \text{if \(i\sigma = 1\)}
    \end{array}
  \right.
\]
for \(\sigma : \myicat{C}(1, c)\). Then by definition \(p(a_{0},
   a_{1}, 0) = a_{0}\) and \(p(a_{0}, a_{1}, 1) = a_{1}\).
\end{proof}

\section{A Failure of Propositional Resizing in Cubical Assemblies}
\label{sec:orga3e7717}
An \emph{assembly}, also called an \emph{\(\omega\)-set}, is a set \(A\) equipped
with a non-empty set \(E_{A}(a)\) of natural numbers for every \(a
  \in A\). When \(n \in E_{A}(a)\), we say \(n\) is a \emph{realizer} for
\(a\) or \(n\) \emph{realizes} \(a\). A morphism \(f : A \to B\) of
assemblies is a function \(f : A \to B\) between the underlying sets
such that there exists a partial recursive function \(e\) such that,
for any \(a \in A\) and \(n \in E_{A}(a)\), the application \(en\) is defined and
belongs to \(E_{B}(f(a))\). In that case we say \(f\) is \emph{tracked}
by \(e\) or \(e\) is a \emph{tracker} of \(f\). We shall denote by
\(\myAsm\) the category of assemblies and morphisms of
assemblies. Note that assemblies can be defined in terms of partial
combinatory algebras instead of natural numbers and partial
recursive functions \cite{vanoosten2008realizability}, and that the
rest of this section works for assemblies on any non-trivial partial
combinatory algebra.

The category \(\myAsm\) is a model of dependent type
theory. Contexts are interpreted as assemblies. Types \(\Gamma
  \vdash A\) are interpreted as families of assemblies \((A(\gamma)
  \in \myAsm)_{\gamma \in \Gamma}\) indexed over the underlying set of
\(\Gamma\). Terms \(\Gamma \vdash a : A\) are interpreted as
sections \(a \in \prod_{\gamma \in \Gamma}A(\gamma)\) such that
there exists a partial recursive function \(e\) such that, for any
\(\gamma \in \Gamma\) and \(n \in E_{\Gamma}(\gamma)\), the
application \(en\) is defined and belongs to
\(E_{A(\gamma)}(a(\gamma))\). For a type \(\Gamma \vdash A\), the
context extension \(\Gamma.A\) is interpreted as an assembly
\((\sum_{\gamma \in \Gamma}A(\gamma), (\gamma, a) \mapsto
  \{\mytuple{n, m} \mid n \in E_{\Gamma}(\gamma), m \in
  E_{A(\gamma)}(a)\})\) where \(\mytuple{n, m}\) is a fixed effective
encoding of tuples of natural numbers. It is known that \(\myAsm\)
supports dependent product types, dependent sum types, extensional
identity types, unit type, disjoint finite coproducts and natural
numbers. See, for example,
\cite{vanoosten2008realizability,longo1991constructive,jacobs1999categorical}.
For a family of assemblies \(A\) over \(\Gamma\), the propositional
truncation \(\|A\|\) is the family
\[
  \|A\|(\gamma) = \left\{
      \begin{array}{lr}
        \{{*}\} & \text{if \(A(\gamma) \neq \emptyset\)} \\
        \emptyset & \text{if \(A(\gamma) = \emptyset\)}
      \end{array}
    \right.
\]
with realizers \(E_{\|A\|(\gamma)}(*) = \bigcup_{a \in
  A(\gamma)}E_{A(\gamma)}(a)\).

It is also well-known that \(\myAsm\) has an impredicative universe
\(\myPER\). It is an assembly whose underlying set is the set of
partial equivalence relations, namely symmetric and transitive relations,
on \(\myN\) and the set of realizers of \(R\) is \(E_{\myPER}(R) =
  \{0\}\). The type \(\myPER \vdash \myel_{\myPER}\) is defined as
\(\myel_{\myPER}(R) = \myN/R\), the set of \(R\)-equivalence classes on
\(\{n \in \myN \mid R(n, n)\}\) with realizers \(E_{\myN/R}(\xi) =
  \xi\). The universe \(\myPER\) classifies \emph{modest families}. An
assembly \(A\) is said to be \emph{modest} if \(E_{A}(a)\) and
\(E_{A}(a')\) are disjoint for distinct \(a, a' \in A\). By definition \(\myN/R\) is
modest for every \(R \in \myPER\). Conversely, for a modest assembly \(A\),
one can define a partial equivalence relation \(R\) such that \(A
  \cong \myN/R\). For the impredicativity of \(\myPER\), see
\cite{hyland1988small,longo1991constructive,jacobs1999categorical}.

The category \(\myAsm\) satisfies the hypotheses of
\cref{org611c4ef} with impredicative universe \(\myPER\), propositional
universe \(\myBoolean\) and the internal category
\(\myicat{B}_{\myord}\) defined in
\cref{sec:orgb20ba66}. We will refer to the presheaf model of
cubical type theory generated by these structures as the \emph{cubical
assembly model}.

\subsection{Propositional Resizing}
\label{sec:org7cd507b}
In cubical type theory, a type \(\Gamma \vdash A\) is a
\emph{homotopy proposition} if the type \(\Gamma, a_{0}, a_{1} : A
   \vdash \myPath(A, a_{0}, a_{1})\) has an inhabitant. For a universe
\(\myUniverse\), we define the universe of homotopy
propositions as
\[\myhProp_{\myUniverse} := \sum_{A : \myUniverse}\prod_{a_{0},
   a_{1} : A}\myPath(A, a_{0}, a_{1}).\]
Following the HoTT book \cite{hottbook}, we regard \(\myhProp_{\myUniverse}\) as a
subtype of \(\myUniverse\).

The \emph{propositional resizing axiom} {\cite[Section 3.5]{hottbook}}
asserts that, for nested universes \(\myUniverse : \myUniverse'\), the
inclusion \(\myhProp_{\myUniverse} \to \myhProp_{\myUniverse'}\) is an
equivalence. When \(\myUniverse\) is an impredicative universe, we
define
\begin{align*}
  A : \myhProp_{\myUniverse'}
  &\vdash A^{*} := \prod_{X : \myhProp_{\myUniverse}}(A \to X) \to X :
    \myhProp_{\myUniverse} \\
  A : \myhProp_{\myUniverse'}
  & \vdash \eta_{A} := \lambda a.\lambda X f.fa : A \to A^{*}.
\end{align*}
If \(\eta_{A}\) is an equivalence for any
\(A : \myhProp_{\myUniverse'}\), then the inclusion
\(\myhProp_{\myUniverse} \to \myhProp_{\myUniverse'}\) is an
equivalence by univalence. Conversely, if the inclusion
\(\myhProp_{\myUniverse} \to \myhProp_{\myUniverse'}\) is an
equivalence, then one can find \(A' : \myhProp_{\myUniverse}\) and
\(e : A \simeq A'\) from \(A : \myhProp_{\myUniverse'}\). Then we have
a function \(\lambda \alpha.e^{-1}(\alpha A' e) : A^{*} \to A\), and
thus \(\eta_{A}\) is an equivalence because both \(A\) and \(A^{*}\)
are homotopy propositions. Note that the construction
\(A \mapsto (A^{*}, \eta_{A})\) works for any homotopy proposition
\(A\) and is independent of the choice of the upper universe
\(\myUniverse'\). Therefore, we can formulate the propositional
resizing axiom in cubical type theory with an impredicative universe
as follows.

\begin{axiom}
  \label{axm:prop-resizing}
  For every homotopy proposition \(\Gamma \vdash A\), the function
  \(\Gamma \vdash \eta_{A} : A \to A^{*}\) is an equivalence.
\end{axiom}

We will show that the cubical assembly model does not satisfy
\cref{axm:prop-resizing}.

\begin{remark}
  We focus on resizing propositions into the impredicative
  universe. The cubical assembly model also has predicative universes,
  assuming the existence of Grothendieck universes in the
  metatheory. It remains an open question whether the predicative
  universes in the cubical assembly model satisfy the propositional
  resizing axiom.
\end{remark}

\subsection{Uniform Objects}
\label{sec:org6dce379}

The key idea to a counterexample to propositional resizing is the
orthogonality of modest and \emph{uniform} assemblies
\cite{vanoosten2008realizability}: if \(X\) is modest and \(A\) is
uniform and well-supported, then the map \(\lambda xa.x : X \to (A \to
X)\) is an isomorphism. Since the impredicative universe \(\myPER\)
classifies modest assemblies, \(\prod_{X : \myPER}(A \to X) \to X\) is
always inhabited for a uniform, well-supported assembly \(A\). We
extend the notion of uniformity for internal presheaves in
\(\myAsm\).

An assembly \(A\) is said to be \emph{uniform} if
\(\bigcap_{a \in A}E_{A}(a)\) is non-empty. We say an internal
presheaf \(A\) on an internal category \(\myicat{C}\) is
\emph{uniform} if every \(A(c)\) is uniform. An internal presheaf
\(A\) on \(\myicat{C}\) is said to be \emph{well-supported} if the
unique morphism into the terminal presheaf is regular epi. For an
internal presheaf \(A\), the following are equivalent:
\begin{itemize}
\item \(A\) is well-supported;
\item \(\| A \|\) is the terminal presheaf;
\item there exists a partial recursive function \(e\) such that, for
any \(c \in \myicat{C}_{0}\) and \(n \in
     E_{\myicat{C}_{0}}(c)\), there exists an \(a \in A(c)\) such that
\(en\) is defined and belongs to \(E_{A}(a)\).
\end{itemize}

By definition a modest assembly cannot distinguish elements with a
common realizer, while elements of a uniform assembly have a common
realizer. Thus a modest assembly ``believes a uniform assembly has at
most one element''. Formally, the following proposition holds.

\begin{proposition}
\label{org43deb71}
Let \(\myicat{C}\) be a category in \(\myAsm\). For a uniform
internal presheaf \(A\) on \(\myicat{C}\) and an internal functor
\(X : \myicat{C}^{\myop} \to \myPER\), the precomposition function
\[i^{*} : (\| A \| \to X) \to (A \to X)\]
is an isomorphism, where \(i : A \to \|A\|\) is the constructor for
propositional truncation. In particular, if, in addition, \(A\) is
well-supported, then the function \(\lambda xa.x : X \to (A \to
   X)\) is an isomorphism.
\end{proposition}

\begin{proof}
Since \(i\) is regular epi, \(i^{*}\) is a monomorphism. Hence it
suffices to show that \(i^{*}\) is regular epi.
Let \(k_{c}\) denote a common realizer of \(A(c)\), namely \(k_{c}
   \in \bigcap_{a \in A(c)}E(a)\).  Let \(c \in \myicat{C}_{0}\) be an
object and \(x : \myyoneda c \times A \to X\) a morphism of
presheaves tracked by \(e\). We have to show that there exists a
morphism \(\hat{x} : \myyoneda c \times \| A \| \to X\) such
that \(\hat{x} \circ (\myyoneda c \times i) = x\) and that a
tracker of \(\hat{x}\) is computable from the code of \(e\). For any
\(\sigma : c' \to c\) and \(a, a' \in A(c')\), we have \(enk_{c'}
   \in E(x(\sigma, a)) \cap E(x(\sigma, a'))\) for some \(n \in
   E(\sigma)\). Since \(X(c')\) is modest, we have \(x(\sigma, a) = x(\sigma,
   a')\). Hence \(x\) induces a morphism of presheaves \(\hat{x} :
   \myyoneda c \times \| A \| \to X\) tracked by \(e\) such that
\(\hat{x} \circ (\myyoneda c \times i) = x\).
\end{proof}

\begin{theorem}
\label{orgbf23d34}
Let \(\Gamma \vdash A\) be a type in the cubical assembly
model. Suppose that \(A\) is uniform and well-supported as an
internal presheaf on \(\mycatEl(\Gamma)\) and does not have a
section. Then the function \(\Gamma \vdash \eta : A \to A^{*}\) is
not an equivalence.
\end{theorem}

\begin{proof}
By \cref{org43deb71}, we see that \(A^{*} =
   \prod_{X : \myhProp}(A \to X) \to X\) has an inhabitant while
\(A\) does not have an inhabitant by assumption.
\end{proof}

\begin{theorem}
\label{org5a648cb}
Let \(\Gamma \vdash A\) be a type in \(\myAsm\). Suppose that \(A\)
is uniform and well-supported but does not have a section. Then the
function \(\Delta \Gamma \vdash \eta : \nabla A \to (\nabla
   A)^{*}\) is not an equivalence.
\end{theorem}

\begin{proof}
By \cref{orgbf23d34}, it suffices to show that
the type \(\Delta \Gamma \vdash \nabla A\) is uniform and well-supported
but does not have a section. For the uniformity, let \(k_{\gamma}\)
be a common realizer of \(A(\gamma)\) for \(\gamma \in
   \Gamma\). For any object \(c \in \myicat{C}\) and element
\(\gamma \in \Gamma\), the code of the constant function \(n
   \mapsto k_{\gamma}\) is a common realizer of \(\nabla A(c, \gamma)
   = \myicat{C}(1, c) \to A(\gamma)\).

For the well-supportedness, let \(e\) be a partial recursive function
such that, for any \(\gamma\) and \(n \in E_{\Gamma}(\gamma)\),
there exists an \(a \in A(\gamma)\) such that \(en\) is defined and
belongs to \(E_{A(\gamma)}(a)\). Then the function \(f\) mapping
\((n, x)\) to the code of the function \(y \mapsto ex\) realizes
that \(\nabla A\) is well-supported. Indeed, for any \(c \in
   \myicat{C}\), \(n \in E_{\myicat{C}}(c)\), \(\gamma \in \Gamma\)
and \(x \in E_{\Gamma}(\gamma)\), the code \(f(n, x)\) realizes the constant
function \(\myicat{C}(1, c) \ni \sigma \mapsto a \in A(\gamma)\)
for some \(a \in A(\gamma)\) such that \(ex \in E_{A(\gamma)}(a)\).

Finally \(\nabla A\) does not have a section because \(\nabla A(1)
   \cong A\) and \(A\) does not have a section.
\end{proof}

\subsection{The Counterexample}
\label{sec:org09fc1c9}
We define an assembly \(\Gamma\) to be \((\myN, n \mapsto \{m \in
   \myN \mid m > n\})\) and a family of assemblies \(A\) on \(\Gamma\)
as \(A(n) = (\{m \in \myN \mid m > n\}, m \mapsto \{n,
   m\})\). Then \(A\) is uniform because every \(A(n)\) has a common
realizer \(n\). The identity function realizes that \(A\) is
well-supported.  To see that \(A\) does not have a section, suppose
that a section \(f \in \prod_{n \in \Gamma}A(n)\) is tracked by a
partial recursive function \(e\). Then for any \(m > n\), we have
\(em \in \{n, f(n)\}\). This implies that \(m \leq e(m + 1) \leq
   f(0)\) for any \(m\), a contradiction. Note that this construction
of \(\Gamma \vdash A\) works for any non-trivial partial
combinatory algebra \(C\) because natural numbers can be effectively encoded in
\(C\).

Since \(\myicat{B}_{\myord}(1, \myI) \cong \myBoolean\) only contains
end-points, the type \(\Delta \Gamma \vdash \nabla A\) in the
cubical assembly model is a fibration and homotopy proposition by
\cref{org242a9f3,org7b84af3}, while by
\cref{org5a648cb} the function \(\Delta \Gamma
   \vdash \eta : \nabla A \to (\nabla A)^{*}\) is not an
equivalence. Hence the propositional resizing axiom fails in the
cubical assembly model.

\section{Conclusion and Future Work}
\label{sec:org2948b31}

We have formulated the axioms for modeling cubical type theory in an
elementary topos given by Orton and Pitts \cite{orton2016axioms} in a
weaker setting and explained how to construct a model of cubical
type theory in a category satisfying those axioms. As a striking
example, we have constructed a model of cubical type theory with an
impredicative and univalent universe in the category of cubical
assemblies which is not an elementary topos. It has turned out that
this impredicative universe in the cubical assembly model does not
satisfy the propositional resizing axiom.

There is a natural question: can we construct a model of type theory
with a univalent and impredicative universe satisfying the
propositional resizing axiom? One possible approach to this question
is to consider a full subcategory of the category of cubical
assemblies in which every homotopy proposition is equivalent to some
modest family. Benno van den Berg \cite{vandenberg2018univalent}
constructed a model of a variant of homotopy type theory with a
univalent and impredicative universe of \(0\)-types that satisfies the propositional
resizing axiom. Roughly speaking he uses a category of degenerate
trigroupoids in the category of \emph{partitioned assemblies}
\cite{vanoosten2008realizability}, and thus the category of cubical
partitioned assemblies is a candidate for such a full
subcategory. However, the model given in
\cite{vandenberg2018univalent} only supports weaker forms of identity
types and dependent product types, and it is unclear whether it can
be seen as a model of ordinary homotopy type theory.

Higher inductive types are another important feature of homotopy type
theory. One can construct some higher inductive types including
propositional truncation in the cubical assembly model
\cite{swan2019church}, internalizing the construction of higher
inductive types in cubical sets \cite{coquand2018higher} using
\(W\)-types with reductions \cite{swan2018wtypes}. An open question,
raised by Steve Awodey, is whether these higher inductive types are
equivalent to their impredicative encodings.

The cubical assembly model is a realizability-based model of type
theory with higher dimensional structures, but it does not seem to be
what should be called a \emph{realizability \(\infty\)-topos}, a
higher dimensional analogue of a realizability topos
\cite{vanoosten2008realizability}. One problem is that, in the cubical
assembly model, realizers seem to play no role in its internal cubical
type theory, because the existence of a realizer of a homotopy
proposition does not imply the existence of a section of it. Indeed,
the cubical assembly model does not satisfy Church's Thesis
\cite{swan2019church} which holds in the effective topos
\cite{hyland1982effective}. One can nevertheless find a left exact
localization of the cubical assembly model in which Church's Thesis
holds \cite{swan2019church}.

Our construction of models of cubical type theory is a syntactic one
following Orton and Pitts \cite{orton2016axioms}. The original idea of
using the internal language of a topos to construct models of
cubical type theory was proposed by Coquand
\cite{coquand2015internal}. There are also semantic and categorical
approaches. Frumin and van den Berg \cite{frumin2018homotopy} presented
a way of constructing a model structure on a full subcategory of an
elementary topos with a path connection algebra, which is
essentially same as the model structure on the category of fibrant
cubical sets described by Spitters \cite{spitters2016cubical}.  Since
they make no essential use of subobject classifiers, we conjecture
that one can construct a model structure on a full subcategory of a
suitable locally cartesian closed category with a path connection
algebra. Sattler \cite{sattler2017equivalence}, based on his earlier
work with Gambino \cite{gambino2017frobenius}, gave a construction of
a right proper combinatorial model structure on a suitable category
with an interval object. Although Gambino and Sattler use Garner's
small object argument \cite{garner2009understanding}
which requires the cocompleteness of underlying categories, their
construction is expected to work for non-cocomplete categories such
as the category of cubical assemblies using Swan's small object
argument over codomain fibrations
\cite{swan2018lifting,swan2018wtypes}.

\label{sec-7}
\bibliography{references}

\label{sec-8}
\appendix
\section{Details of Composition for Gluing and Universe}
\label{sec:org2687926}
We give explicit definitions of composition operations for gluing
and universes described in \cref{sec:org7bb1bad}.

Before that, we introduce some notations. for a fibration \(\Gamma,
  i : \myI \vdash A(i)\), one can derive the \emph{composition operation}
\begin{mathpar}
  \inferrule
  {\Gamma \vdash e : \myBoolean \\
    \Gamma \vdash \varphi : \myCof \\
    \Gamma, i : \myI \vdash f(i) : \varphi \to A(i) \\
    \Gamma \vdash a : A(e) \\
    \Gamma, u : \varphi \vdash f(e)u = a}
  {\Gamma \vdash \mycomp_{e}^{i}(A(i), f(i), a) : A(\bar{e})}
\end{mathpar}
such that \(\Gamma, u : \varphi \vdash f(\bar{e})u =
  \mycomp_{e}^{i}(A(i), f(i), a)\). Concretely, for a fibration
structure \(\alpha : \myFib(A)\), we define
\[\gamma : \Gamma \vdash \mycomp_{e}^{i}(A(i), f(i), a) :=
  \alpha(\lambda i.(\gamma, i), e, \varphi, \lambda ui.f(i)u, a).\]
In the notation \(\mycomp_{e}^{i}(A(i), f(i), a)\), the variable
\(i\) is considered to be bound. Usually we use the composition
operation in the form of
\[\mycomp_{e}^{i}(A(i), [(u_{1} : \varphi_{1}) \mapsto g_{1}(u_{1},
  i), \dots, (u_{n} : \varphi_{n}) \mapsto g_{n}(u_{n}, i)], a)\]
with a system \([(u_{1} : \varphi_{1}) \mapsto g_{1}(u_{1}, i),
  \dots, (u_{n} : \varphi_{n}) \mapsto g_{n}(u_{n}, i)] : \varphi_{1}
  \lor \dots \lor \varphi_{n} \to A(i)\).

\subsection{Some Derived Notions and Operations}
\label{sec:org2351daa}

We recall some notions and operations derivable in cubical type
theory without gluing and universes.

Composition operations are preserved by function application
{\cite[Section 5.2]{cohen2016cubical}}: one can derive an operation
\begin{mathpar}
  \inferrule
  {\Gamma, i : \myI \vdash h(i) : A(i) \to B(i) \\
    \Gamma \vdash e : \myBoolean \\
    \Gamma \vdash \varphi : \myCof \\
    \Gamma, i : \myI \vdash f(i) : \varphi \to A(i) \\
    \Gamma \vdash a : A(e) \\
    \Gamma, u : \varphi \vdash f(e)u = a}
  {\Gamma \vdash \myprsv_{e}^{i}(h(i), f(i), a) : \myPath(B(\bar{e}), c_{1}, c_{2})}
\end{mathpar}
such that \(\Gamma, u : \varphi, j : \myI \vdash
   h(\bar{e})(f(\bar{e})u) = \myprsv_{e}^{i}(h(i), f(i), a)j\), where
\(c_{1} = \mycomp_{e}^{i}(B(i), h(i) \circ f(i), h(e)a)\) and
\(c_{2} = h(\bar{e})(\mycomp_{e}^{i}(A(i), f(i), a))\).

Equivalences are characterized by a kind of extension property
{\cite[Section 5.3]{cohen2016cubical}}: for fibrations \(\Gamma \vdash A\) and \(\Gamma \vdash
   B\), one can derive an operation
\begin{mathpar}
  \inferrule
  {\Gamma \vdash f : A \simeq B \\
    \Gamma \vdash e : \myBoolean \\
    \Gamma \vdash \varphi : \myCof \\
    \Gamma \vdash b : B \\
    \Gamma \vdash p : \varphi \to \sum_{a : A}\myPath(B, b, fa)}
  {\Gamma \vdash \myequiv(f, p, b) : \sum_{a : A}\myPath(B, b, fa)}
\end{mathpar}
such that \(\Gamma, u : \varphi \vdash pu = \myequiv(f, p, b)\).

For a fibration \(\Gamma, i : \myI \vdash A(i)\), we define a
function called \emph{transport} \(\Gamma, e : \myBoolean \vdash
   \mytransport_{e}^{i}(A(i)) : A(e) \to A(\bar{e})\) to be
\(\mytransport_{e}^{i}(A(i))a = \mycomp_{e}^{i}(A(i), [], a)\).
This function \(\mytransport_{e}^{i}(A(i))\) is an equivalence
{\cite[Section 7.1]{cohen2016cubical}}.

\subsection{Gluing}
\label{sec:org07f4ccc}

\begin{proof}
[Proof of \cref{orgeaa893e}]
Let
\(p : \myinterval \to \Gamma\),
\(e : \myBoolean\),
\(\psi : \myCof\),
\(g : \psi \to \prod_{i : \myinterval}\prod_{u : \varphi(pi)}A(u)\),
\(h : \psi \to \prod_{i : \myinterval}B(pi)\),
\(a : \prod_{u : \varphi(pe)}A(u)\) and
\(b : B(pe)\),
and suppose
\(\forall_{v : \psi}\forall_{i : \myinterval}\forall_{u :
   \varphi(pi)}f(u)(gviu) = hvi\),
\(\forall_{u : \varphi(pe)}f(u)(au) = b\) and
\(\forall_{v : \psi}gve = a \land hve = b\).
We have to find elements \(\bar{a} : \prod_{u :
   \varphi(p\bar{e})}A(u)\) and \(\bar{b} : B(p\bar{e})\) such that
\(\forall_{u : \varphi(p\bar{e})}f(u)(\bar{a}u) = \bar{b}\) and
\(\forall_{v : \psi}gv\bar{e} = \bar{a} \land hv\bar{e} = \bar{b}\).
We define
\begin{align*}
  \bar{b}_{1} &:= \mycomp_{e}^{i}(B(pi), [(v : \psi) \mapsto hvi], b) :
                B(p\bar{e}) \\
  \delta &:= \forall_{i : \myinterval}\varphi(pi) : \myCof \\
  \bar{a}_{1} &:= \lambda w.\mycomp_{e}^{i}(A(wi), [(v : \psi) \mapsto gvi(wi)], a(we)) :
                \prod_{w : \delta}A(w\bar{e}) \\
  q &: \prod_{w : \delta}\myPath(\bar{b}_{1}, f(w\bar{e})(\bar{a}_{1}w)) \\
  qw &:= \myprsv_{e}^{i}(f(wi), [(v : \psi) \mapsto gvi(wi)], a(we)) \\
  \bar{a} &: \prod_{u : \varphi(p\bar{e})}A(u) \\
  q_{2} &: \prod_{u : \varphi(p\bar{e})}\myPath(\bar{b}_{1}, f(u)(\bar{a}u)) \\
  (\bar{a}u, q_{2}u) &:= \myequiv(f(u), [(w : \delta) \mapsto (\bar{a}_{1}w, qw), (v : \psi) \mapsto (gv\bar{e}u, \lambda i.\bar{b}_{1})], \bar{b}_{1}) \\
  \bar{b} &:= \mycomp_{0}^{i}(B(p\bar{e}), [(u : \varphi(p\bar{e})) \mapsto q_{2}ui, (v : \psi) \mapsto hv\bar{e}], \bar{b}_{1}) :
            B(p\bar{e})
\end{align*}
Then one can derive that \(\bar{b} = q_{2}u1 = f(u)(\bar{a}u)\) for
\(u : \varphi(p\bar{e})\) and that \(\bar{a} = gv\bar{e}\) and
\(\bar{b} = hv\bar{e}\) for \(v : \psi\). Moreover, for every \(w :
   \prod_{i : \myI}\varphi(pi)\), we have \(\bar{a}(w\bar{e}) =
   \bar{a}_{1}w = \mycomp_{e}^{i}(A(wi), [(v : \psi) \mapsto gvi(wi)],
   a(we))\) which means the preservation of fibration structure by the
function \(\Gamma, u : \varphi \vdash \lambda (a, b).au :
   \myGlue(\varphi, f) \to A(u)\).
\end{proof}

\subsection{Universes}
\label{sec:org07d573e}

\begin{proof}
[Proof of \cref{org97bc451}]
Let \(e : \myBoolean\), \(\varphi : \myCof\), \(f : \varphi \to
   \myinterval \to \myUniverse^{F}\) and \(B : \myUniverse^{F}\) such
that \(\forall_{u : \varphi}fue = B\). We have to find a \(\bar{B} :
   \myUniverse^{F}\) such that \(\forall_{u : \varphi}fu\bar{e} =
   \bar{B}\). Let \(A := \lambda u.fu\bar{e} : \varphi \to
   \myUniverse^{F}\). We have an
equivalence \(g := \lambda u.\mytransport_{\bar{e}}^{i}(fui) : \prod_{u : \varphi}Au \simeq B\). Let \(\bar{B} :=
   \mySGlue(\varphi, g) : \myUniverse^{F}\), then \(\forall_{u :
   \varphi}fu\bar{e} = Au = \bar{B}\).
\end{proof}
\end{document}